\documentclass[smallextended]{svjour3}
\smartqed
\usepackage{amsfonts}
\usepackage{mathtools}
\usepackage[a4paper,margin=1in]{geometry}
\usepackage{amsmath,amssymb,amsfonts}
\usepackage{graphicx}

\usepackage{amsthm}
\usepackage{enumitem}
\usepackage{booktabs}
\usepackage{algorithm}
\usepackage{algorithmic}
\usepackage{float}
\usepackage{xcolor}

\usepackage[round,authoryear]{natbib}
\usepackage[justification=raggedright,singlelinecheck=false,font=small,labelfont=bf]{caption}
\usepackage[colorlinks=true,linkcolor=blue,citecolor=blue,urlcolor=blue]{hyperref}
\usepackage[T1]{fontenc}
\usepackage{lmodern}

\theoremstyle{plain}

\newtheorem*{theorem*}{Theorem}

\title{Functional Time Series Forecasting of Distributions: A Koopman-Wasserstein Approach}
\author{Ziyue Wang \and Yuko Araki}

\institute{yuko.araki.c5@tohoku.ac.jp}   
\date{}

\begin{document}
\maketitle

\begin{abstract}
  \noindent
We present a novel method for forecasting the temporal evolution of probability distributions observed at discrete time points. Building on the Dynamic Probability Density Decomposition (DPDD), we incorporate distributional dynamics into Wasserstein geometry using a Koopman operator framework. Our method introduces an importance-weighted variant of Extended Dynamic Mode Decomposition (EDMD), allowing for accurate, closed-form forecasts in 2-Wasserstein space. We establish theoretical guarantees, demonstrating that our estimator achieves spectral convergence and an optimal Wasserstein error. Simulation studies and a real-world application to U.S. housing price distributions reveal significant improvements over existing methods, such as Wasserstein Autoregression. By integrating optimal transport, functional time series modeling, and spectral operator theory, DPDD provides a scalable and interpretable solution for distributional forecasting. This work has broad implications for behavioral science, public health, finance, and neuroimaging—fields where evolving distributions are commonplace. Our framework contributes to functional data analysis on non-Euclidean spaces and serves as a general tool for modeling and forecasting distributional time series.
\end{abstract}


\section{Introduction}\label{sec:intro}
In fields like psychology, education, cognitive science, public health, biostatistics, and the social sciences, time series data increasingly emerge as repeated samples from evolving probability distributions rather than as scalar or vector-valued observations. For instance, in population studies, the distributions of age-at-death or healthy life expectancy can change over time. In behavioral research, longitudinal survey responses create histograms that reflect variations across different time points. In neuroscience and physiology, time-varying signals—such as brain activity, heart rate variability, and glucose levels—indicate shifts in underlying distributions. These distributions are not directly observable; they must be estimated from finite data samples. Summary statistics like means and variances often do not adequately represent the full range of uncertainty, heterogeneity, or behavioral complexity. Therefore, modeling and forecasting the evolution of probability distributions based on observed samples has become a key challenge in applied behavioral, biomedical, and socio-economic research.

This perspective encourages approaches that consider distributions as the primary units of analysis, often by embedding them within a functional or geometric framework. In the realm of functional data analysis (FDA) paradigm, discrete or multivariate observations are typically transformed into smooth functions and analyzed using regression or dimension reduction techniques. For instance, \citet{araki2009functional} developed model selection criteria for functional linear regression, where multivariate observations are smoothed into functional predictors. Although their work does not focus on distributions, it illustrates the broader strategy of elevating finite data into function space before conducting statistical modeling—a viewpoint that also underpins the current study.

When probability distributions are represented in Wasserstein space, which is equipped with the 2-Wasserstein metric from optimal transport theory, the resulting space exhibits a nonlinear geometric structure. A common approach involves linearizing the Wasserstein space at a reference distribution $\mu_\oplus$, apply FDA techniques in the tangent space, and then mapping predictions back using the exponential map. This paradigm underlies methods such as Wasserstein autoregression (WAR) \citep{chen2023wasserstein} and Fréchet regression \citep{petersen2019frechet}, which are effective in one-dimensional settings.

Tangent-space-based methods encounter three main limitations. First, WAR offers only one-step-ahead forecasts and does not have a closed-form representation for continuous-time dynamics. Second, the reliance on a finite number of principal components for linearization can result in information loss, particularly in the distribution tails. Third, for multivariate distributions ($d > 1$), the lack of a canonical Hilbert structure and explicit quantile maps complicates the extension of tangent-space formulations.

To tackle these challenges, we build upon the Dynamic Probability Density Decomposition (DPDD) framework introduced by \citet{zhao2022dpdd}, which models distributional dynamics through the spectral decomposition of the Koopman operator. Although the original DPDD offered a promising structure for forecasting, it lacked a formal geometric interpretation and theoretical guarantees. In this work, we extend DPDD both theoretically and algorithmically. Theoretically, we demonstrate that the density ratio dynamics define a coordinate chart on the Wasserstein manifold, situating distributional evolution within the intrinsic geometry of optimal transport. Algorithmically, we enhance the estimation procedure by developing an importance-weighted extension of Extended Dynamic Mode Decomposition (EDMD)—a method for approximating the Koopman operator originally proposed by \citet{williams2015edmd}.

We establish theoretical guarantees for our method: the estimated Koopman operator converges at a rate $\mathcal{O}(M^{-1/2})$, where $M$ represents the length of the training trajectory. Additionally, the corresponding density forecasts achieve a mean-squared 2-Wasserstein error of order $\mathcal{O}(M^{-1/2})$ under mild assumptions. These findings offer a statistical foundation for applying Koopman spectral theory to distributional functional data.

Simulation studies and a real-world application to U.S. housing price distributions highlight the accuracy, robustness, and interpretability of the extended DPDD method. By treating evolving distributions as function-valued time series within a geometric framework, this research advances the development of functional data analysis for non-Euclidean and distributional data.

The remainder of this paper is organized as follows: 
Section~\ref{sec:prelim} provides background on Wasserstein geometry, stochastic dynamics, and Koopman operator theory. 
Section~\ref{sec:dpdd} introduces our theoretical extension of DPDD and presents the proposed estimation procedure using importance-weighted EDMD. 
Section~\ref{sec:implementation} describes the practical implementation details. 
Sections~\ref{sec:simulation} and~\ref{sec:app} report the simulation results and empirical analysis of the U.S. housing price distributions. 
Finally, Section~\ref{sec:discussion} concludes the paper with a discussion of contributions, limitations, and directions for future research.

\section{Preliminaries}\label{sec:prelim}
This section provides the theoretical foundation for the development of our forecasting method. We first review the geometry of the 2-Wasserstein space and then introduce key concepts from stochastic processes and Koopman operator theory. We conclude with a presentation of weighted Extended Dynamic Mode Decomposition (EDMD), which forms the basis of our estimator.

\subsection{The 2-Wasserstein Space}
Let $(\mathcal{X}, \|\cdot\|)$ be a Polish space and let $\mathcal{P}_2(\mathcal{X})$ denote the set of Borel probability measures on $\mathcal{X}$ with finite second moments:
\[
\mathcal{P}_2(\mathcal{X}) := \left\{ \mu \in \mathcal{P}(\mathcal{X}) : \int_{\mathcal{X}} \|x\|^2 \, d\mu(x) < \infty \right\}.
\]
The 2-Wasserstein distance between $\mu, \nu \in \mathcal{P}_2(\mathcal{X})$ is defined as
\[
W_2(\mu, \nu) := \left( \inf_{\pi \in \Pi(\mu, \nu)} \int_{\mathcal{X} \times \mathcal{X}} \|x - y\|^2 \, d\pi(x, y) \right)^{1/2},
\]
where $\Pi(\mu, \nu)$ is the set of all couplings of probability measures $\mu$ and $\nu$.

If $\mu$ is absolutely continuous, there exists a Brenier potential $\varphi$ and a unique optimal transport map $T^* = \nabla \varphi$ pushing $\mu$ to $\nu$, according to Brenier's theorem, see \citet{villani2009optimal}.
As demonstrated by \citet{ambrosio2008gradient} and \citet{bigot2017central}, the basic concepts of Riemannian manifolds can also be generalized to the 2-Wasserstein space $\mathcal{P}_2(\mathcal{X})$ using the metric 2-Wasserstein distance. For a measurable map $T\!:D\!\to\!D$ and a Borel measure $\mu_\ast$ on $D$, the \emph{push forward} measure is defined as 
\[
  T_{\sharp}\mu_\ast(A)\;=\;\mu_\ast\bigl(\{\,r\in D:\,T(r)\in A\,\}\bigr), 
  \qquad\forall\,A\in\mathcal B(D).
\]
 This gives rise to the exponential and logarithmic maps:
\[
\log_{\mu}(\nu) = \nabla \varphi - \mathrm{Id}, \quad \exp_{\mu}(v) = (\mathrm{Id} + v)_\# \mu,
\]
which locally linearizes $\mathcal{P}_2(\mathcal{X})$ at $\mu$ through its tangent space $T_\mu \mathcal{P}_2$. Here, $Id : \mathcal X \to \mathcal X$ denotes  the \emph{identity map}; $Id(x)=x$; and $T_\mu \mathcal{P}_2$ denotes the tangent space of $\mathcal{P}_2$ at $\mu$, which is a Hilbert space.

The 2 Wasserstein space $\mathcal P_{2}$ is not linear. Hence, treating it as a Hilbert space is theoretically inaccurate (\citet{villani2009optimal}). However, its tangent space $T_{\mu}\mathcal P_{2}$ is a Hilbert space, and several recent studies exploit the resulting infinite dimensional manifold structure by projecting distributions onto $T_{\mu}\mathcal P_{2}$ with exponential and logarithmic maps and using parallel transport along the tangent bundle to formulate regression models, for example, \citet{zhang2022wasserstein} and \citet{chen2023wasserstein}.  Nevertheless, in multivariate settings ($d\!\ge\!2$), there is generally no global isometry from $\mathcal P_{2}(\mathbb R^{d})$ to any fixed tangent space and even exponential and logarithmic maps admit no closed form expressions.

To overcome this, in the following subsection, we introduce our approach using the dynamics of distributions in the Wasserstein space, attracting dynamic modes utilizing Koopman operator.

\subsection{Stochastic Dynamics and the Koopman Semigroup}
We consider an It\^o diffusion
\[
dX_t = b(X_t) \, dt + \sigma(X_t) \, dW_t, \quad X_0 \sim \mu_0,
\]
\label{eq:sde}
with drift $b : \mathbb{R}^d \to \mathbb{R}^d$ and diffusion $\sigma : \mathbb{R}^d \to \mathbb{R}^{d \times d}$ satisfying Lipschitz conditions, and $W_t$ denotes a $d$ dimensional standard Brownian motion. The associated Fokker--Planck equation describes the evolution of the density $p_t$:
\[
\partial_t p_t = \mathcal{L}^* p_t, \quad \mathcal{L}^* = -\nabla \cdot (b \cdot) + \frac{1}{2} \nabla^2 : (\sigma \sigma^\top \cdot).
\]
here $\mathcal{L}^*$ denotes the Fokker--Planck operator. $\nabla\!\cdot(\,\cdot\,)$ is the divergence acting on the placeholder “$\!\cdot\!$”, $(\sigma\sigma^{\top}\,\cdot)$ means the matrix $\sigma\sigma^{\top}$ multiplies the density $p$, and “$:$” denotes Frobenius double contraction, i.e.\ $\nabla^{2}\!:\!A=\sum_{i,j}\partial_{x_i x_j}A_{ij}$.\\ 
The Koopman semigroup $\{\mathcal{K}^t\}_{t \geq 0}$ acts on observables $f \in L^2(p_s)$ as
\[
\mathcal{K}^t f(x) := \mathbb{E}[f(X_t) \mid X_0 = x],
\]
and its generator is
\[
\mathcal{L} f = \langle b, \nabla f \rangle + \frac{1}{2} \mathrm{Tr}(\sigma \sigma^\top \nabla^2 f).
\]

 The Koopman semigroup provides a linear representation of nonlinear stochastic dynamics on observables, yielding closed form multi step forecasts $c_{t+h}=\exp^{\Lambda h}c_t$. It also avoids reliance on tangent space log/exp maps, facilitating multivariate extensions (see Section 6.2).

\subsection{Weighted Extended Dynamic Mode Decomposition}
 The Koopman operator associated with a dynamical system is, under mild regularity conditions, a linear operator on observables. This enables a linear representation of nonlinear dynamics.
In the proposed framework, we improve the estimation of the Koopman operator by introducing an importance-weighted extension of Extended Dynamic Mode Decomposition (EDMD) (\citet{williams2015edmd}), specifically designed for systems evolving in the 2-Wasserstein space.
This modification addresses three key finite-sample issues in classical EDMD:$(i)$bias from non-uniform sampling of the stationary distribution, $(ii)$ inflated variance in low-density regions, and $(iii)$ numerical instability due to ill-conditioned Gram matrices.

Let $\Psi(x) = (\psi_1(x), \dots, \psi_J(x))^\top$ be a dictionary of the basis functions with $\psi_1 \equiv 1$.  Given a trajectory $\{z_k\}_{k=1}^M$ from the stationary process, we define the importance weights as:
\[
w_k := \frac{\hat{p}_s(z_k)}{\sum_{\ell=1}^M \hat{p}_s(z_\ell)},
\]
where $\hat{p}_s$ denotes a kernel density estimator. Classical EDMD estimates the inner products 
$\langle \psi_i,\psi_j\rangle_{L^2(p_s)}$ 
where $\langle \psi_i, K_{\Delta t}\psi_j\rangle_{L^2(p_s)}$  
with unweighted empirical averages, which in finite samples actually lie
in $L^2(\hat\mu_M)$, where $\hat\mu_M=(1/M)\sum_{k=1}^M\delta_{z_k}$.
When the trajectory $\{z_k\}$ is drawn from stationary diffusion,
$\hat\mu_M$ approximates $p_s$ only at the rate $M^{-1/2}$ and exhibits high variance in regions where $p_s$ is small, leading to a degraded estimation of Gram and cross covariance matrices. 
To recover the target space $L^2(p_s)$, we adopted the importance weights 
$w_k\propto\hat p_s(z_k)$.
Then $\mathbb E[w_k f(z_k)]=\int f(x)p_s(x)\,dx$, so the weighted Gram and cross-covariance matrices are:
\[
G_M := \sum_{k=1}^M w_k \Psi(z_k) \Psi(z_k)^\top, \quad A_M := \sum_{k=1}^M w_k \Psi(z_{k+1}) \Psi(z_k)^\top.
\]
They are unbiased Monte-Carlo estimators of their population analogues
$G_\ast,A_\ast$, where 
$G_\ast := \mathbb{E}_{p_s}[\Psi(X)\Psi(X)^\top]$ and 
$A_\ast := \mathbb{E}_{p_s}[\Psi(X_{t+\Delta t})\Psi(X_t)^\top] $
in their population counterparts.
This removes the systematic bias induced by non-uniform sampling.
reduces the condition number of $G_M$, and yields the
spectral-convergence rate 
$\|\,\widehat K_M-K_\ast\|_2=O_\mathbb P(M^{-1/2})$
established in Theorem \ref{thm:spectral}.

Intuitively, the weights play the same role as importance sampling in
Monte-Carlo integration: they down-weight rarely visited states, whereas
up-weighting states with higher stationary density, leading to a lower
variance, and better finite-sample eigenvector estimates.

The Koopman operator is approximated as
\[
\hat{K}_M := A_M G_M^\dagger,
\]
where $G_M^\dagger$ denotes the Moore–Penrose pseudoinverse, The eigenpairs of $\hat{K}_M$ approximate the leading Koopman modes and form the basis of the proposed DPDD framework.

\section{Model}
\label{sec:dpdd}

In this section, we present the Dynamic Probability Density Decomposition (DPDD) framework. DPDD utilizes the spectral decomposition of the Koopman operator linked to an underlying stochastic dynamical system to create a finite-dimensional modal representation of the density ratio \(q_t = p_t/p_s\) within the weighted Hilbert space \(L^2(p_s)\).

Unlike tangent-space methods such as Wasserstein Autoregression (WAR), DPDD does not rely on the logarithmic or exponential maps of the Wasserstein manifold \(\mathcal{P}_2\), and is thus capable of multi-dimensional cases. 
This choice circumvents several limitations associated with log/exp–based methods: (i) for multivariate cases $(d\geq 2)$ closed-form expressions for log/exp maps are typically unavailable, complicating computations; (ii) local linearization within a tangent space may overlook significant nonlinear geometric features, especially in distribution tails; and (iii) in high-dimensional or non-Gaussian contexts, numerical instability and projection bias can substantially diminish forecast accuracy. In contrast, DPDD directly represents dynamics in density-ratio coordinates through Koopman modes, thereby maintaining the global geometric structure without the need for repeated tangent-space transformations.
The original formulation of DPDD was proposed by \citep{zhao2022dpdd} as a spectral surrogate for diffusion forecasting. However, its relationship to the geometry of the Wasserstein space was not explicitly addressed.

We begin by presenting the modal expansion of stationary density ratios, deriving the weighted EDMD estimator used to approximate the Koopman operator, and finally establishing the main theoretical guarantees regarding convergence and prediction accuracy.

\subsection{Modal Expansion of Stationary Density Ratios}
\label{subsec:dpdd_modal}

Let \(\{X_{t}\}_{t\ge0}\) satisfy the SDE~\eqref{eq:sde} with a stationary density \(p_{\mathrm{s}}\) (Assumption~\ref{B:ergodic}). Define the density ratio \(q_{t}=p_{t}/p_{\mathrm{s}}\). Since \(\int q_{t}\,p_{\mathrm{s}}=1\), we can expand \(q_{t}\) in the orthonormal eigenbasis \(\{\varphi_{i}\}_{i\ge0}\) of the generator \(\mathcal{L}\) in \(L^{2}(p_{\mathrm{s}})\):
\begin{equation}\label{eq:q_expansion}
  q_{t}(x) = 1 + \sum_{i=1}^{\infty} c_{i}(0)\, e^{\lambda_{i}t}\, \varphi_{i}(x), \quad c_{i}(0)=\langle q_{0},\varphi_{i}\rangle_{p_{\mathrm{s}}}.
\end{equation}
This expansion represents the density as a stationary baseline plus a linear combination of Koopman eigenfunctions. Each mode \(\varphi_{i}\) evolves independently at an exponential rate \(\lambda_{i}\), and thus forecasting the future density reduces it to propagating a finite number of modal coefficients. DPDD estimates the leading eigenpairs \(\{(\widehat\lambda_{i},\widehat\varphi_{i})\}_{i=1}^{r}\) from the data.

\subsection{Weighted EDMD Estimator}
\label{subsec:dpdd_edmd}
Given a trajectory \(\{z_{k}\}_{k=1}^{M}\) sampled from the stationary distribution, we construct weighted EDMD matrices:
\begin{equation}\label{eq:GA_mats}
  G_{M} = \sum_{k=1}^{M} w_{k}\,\Psi(z_{k})\Psi(z_{k})^{\mathsf{T}}, \quad
  A_{M} = \sum_{k=1}^{M} w_{k}\,\Psi(z_{k+1})\Psi(z_{k})^{\mathsf{T}}, \quad
  w_{k}\propto\widehat p_{\mathrm{s}}(z_{k}),
\end{equation}
where \(\widehat p_{\mathrm{s}}\) is a kernel density estimator. The Koopman operator is approximated by \(\widehat{\mathcal{K}}_{M}=A_{M}G_{M}^{\dagger}\), and its leading eigenpairs form the spectral basis for the DPDD method.

\subsection{Spectral Convergence}
\label{subsec:spectral_conv}

Let $L$ be the infinitesimal generator of the underlying Markov semigroup and
$K_{\Delta t}\!: = e^{\Delta t L}$ is the Koopman operator at time–step $\Delta t$.
For each eigenpair $(\lambda_i,\varphi_i)$ of $L$ (\,$L\varphi_i=\lambda_i\varphi_i$\,), we define the discrete time Koopman eigenvalue as
\[
  \mu_i \coloneqq e^{\Delta t\lambda_i}, \qquad i=1,\dots,r,
\]
and its finite dimensional projection onto the dictionary
$\Psi(x)=(\psi_1(x),\ldots,\psi_J(x))^{\top}$ by
\[
  \xi_i \coloneqq \Pi_J\varphi_i
  =\bigl(\langle\varphi_i,\psi_1\rangle_{L^{2}(p_s)},\;\ldots,\;
          \langle\varphi_i,\psi_J\rangle_{L^{2}(p_s)}\bigr)^{\!\top}\in\mathbb{R}^J.
\]

In practice, $(\widehat\mu_i,\widehat\xi_i)$ are obtained by first computing the weighted EDMD operator $\widehat K_M = A_M G_M^\dagger$ from the 
weighted Gram and cross-covariance matrices $G_M$ and $A_M$ constructed 
in Section~\ref{subsec:dpdd_edmd}. Then, $\widehat\mu_i$ are considered the 
$r$ dominant eigenvalues of $\widehat K_M$, and $\widehat\xi_i$ as the corresponding 
right eigenvectors, normalized such that $\|\widehat\xi_i\|_2=1$.

Let $(\widehat\mu_i,\widehat\xi_i)_{i=1}^{r}$ denote $r$
dominant right eigenpairs of $\widehat K_M$, normalized such that
$\|\widehat\xi_i\|_2=1$.
Then $\widehat\mu_i\!\to\!\mu_i$ and
$\widehat\varphi_i(x)\!=\!\widehat\xi_i^{\!\top}\Psi(x)\!\to\!\varphi_i(x)$
in probability as $M\to\infty$ using the following theorem.

\begin{theorem}[Spectral Convergence of Weighted EDMD]
\label{thm:spectral}
Assume \textbf{\ref{B:ergodic}--\ref{B:gram}} and \(\sup_{x}\|\Psi(x)\|^{4}<\infty\). Fix \(r\) such that the leading \(r\) Koopman eigenvalues are simple. Then,
\[\max_{1\le j\le r} \left( |\widehat\mu_{j}-\mu_{j}| + \|\widehat\xi_{j}-\xi_{j}\|_{2} \right) = \mathcal{O}_{\mathbb{P}}(M^{-1/2}).\]
Specifically, \(\widehat\lambda_{j}\to\lambda_{j}\) and \(\widehat\varphi_{j}\to\varphi_{j}\) in \(L^{2}(p_{\mathrm{s}})\) for \(j=1,\dots,r\).
\end{theorem}

\begin{proof}[Proof sketch]
Here, we outline the proof strategy:
\begin{enumerate}[label=(\alph*),leftmargin=20pt]
  \item Lemma~\ref{lem:GA_conc} shows \(\|G_{M}-G_{\ast}\|_{2}=O_{\mathbb{P}}(M^{-1/2})\); a similar bound holds for \(A_{M}-A_{\ast}\). Hence, \(\|\widehat{\mathcal{K}}_{M}-\mathcal{K}_{\ast}\|_{2}=O_{\mathbb{P}}(M^{-1/2})\).

  \item Based on the Davis--Kahan theorem, the empirical eigenvalue and eigenvector errors satisfy:
  \[
    |\widehat\mu_{k}-\mu_{k}| \lor \|\widehat\xi_{k}-\xi_{k}\|_{2} = O_{\mathbb{P}}(M^{-1/2}).
  \]

  \item As \(M\to\infty\), the estimated modal system converges to the true Koopman spectrum. Since we choose \(J\ll M^{1/2}\), the Galerkin bias \(J^{-q}\) is asymptotically negligible.
\end{enumerate}
\end{proof}

Theorem \ref{thm:spectral} states that our weighted EDMD
estimate $\widehat K_M$ converges to the true Koopman operator in the spectral
norm at the Monte-Carlo rate \(O_\mathbb P(M^{-1/2})\). This illustrates that the empirical eigenfunctions learned by DPDD provide an increasingly accurate basis for forecasting distributional dynamics.

\subsection{Wasserstein Error Bound}
\label{subsec:w2_error}
Corollary~\ref{cor:w2risk} provides a direct link between the spectral convergence of Theorem~\ref{thm:spectral} and the accuracy of density forecasts under the $W_2$ metric. For a fixed forecast horizon $h$, the $W_2$ forecast risk decays at the optimal Monte Carlo rate $O(M^{-1/2})$, with the most linear growth in $h$. 
This result ensures that improvements in operator estimation directly translate into improvements in the predictive accuracy of the densities. 

 Let \(\widehat p_{T+h}\) be the \(h\)-step DPDD forecast from \(r\) modes and \(p_{T+h}\) be the real density. We have the following corollary:

\begin{corollary}[\(W_2\) Risk]
\label{cor:w2risk}
Under Theorem~\ref{thm:W2-risk}'s assumptions, and for a fixed forecast horizon \(h\),
\[\mathbb{E}\,W_{2}^{2}\bigl(p_{T+h},\widehat p_{T+h}\bigr) = \mathcal{O}(M^{-1/2}).\]
\end{corollary}

Proof appears in Appendix~\ref{app:w2_proof}, where semigroup properties are used to separate modal and spectral errors. 
Corollary~\ref{cor:w2risk} effectively converts the operator error bound to a forecasting error bound in the natural geometry for distributions. 
This confirms that the DPDD framework inherits the same convergence rate from the spectral domain to the prediction domain, thus providing a clear and interpretable guarantee.
\subsection{Main Theorem}
\label{subsec:main_theorem}

Theorem~\ref{thm:main} consolidates the theoretical guarantees of DPDD by demonstrating that it simultaneously achieves the minimax-optimal $O_\mathbb{P}(M^{-1/2})$ rate for both spectral estimation and $W_2$ forecasting accuracy. 
This dual guarantee implies that the accurate recovery of the Koopman spectral structure ensures equally accurate multi-step distributional forecasts under mild dictionary growth conditions.

\begin{theorem}[DPDD Consistency and Prediction Accuracy]
\label{thm:main}
Let Assumptions \textbf{\ref{B:ergodic}--\ref{B:gram}} hold and choose dictionary size \(J\ll M^{1/2}\). For any fixed forecast horizon \(h\ge1\):
\begin{enumerate}[leftmargin=20pt]
  \item (Spectral consistency)
        \(\max_{1\le j\le r} (|\widehat\lambda_{j}-\lambda_{j}| + \|\widehat\varphi_{j}-\varphi_{j}\|_{L^{2}(p_{\mathrm{s}})}) = \mathcal{O}_{\mathbb{P}}(M^{-1/2})\).
  \item (\(W_2\) prediction risk)
        \(\mathbb{E}\,W_{2}^{2}(p_{T+h},\widehat p_{T+h}) = \mathcal{O}(M^{-1/2})\).
\end{enumerate}

\end{theorem}

By integrating these two aspects into a single theorem, we emphasize that DPDD's statistical guarantees are comprehensive. The method is theoretically robust for estimating the latent dynamical structure (Koopman spectrum) and for providing accurate and reliable distribution forecasts in the $W_2$ metric.
This makes the framework especially suitable for practical applications that require both interpretability and predictive performance.

The next section empirically validates these results and compares DPDD with existing methods such as WAR.

\section{Practical Implementation of DPDD}
\label{sec:implementation}

This section presents the implementation steps of the proposed Dynamic Probability Density Decomposition (DPDD) method. The algorithm assumes the availability of a stationary trajectory from the underlying diffusion, a set of empirical distributions over time, and a suitable function basis.

\subsection{Required Components}

This method relies on the following inputs:

\begin{itemize}
  \item A trajectory $\{z_k\}_{k=1}^M$ sampled from the stationary distribution $p_s$ of the stochastic process. This is used to estimate Koopman eigenfunctions.
  \item A sequence of empirical distributions $\{\mathcal{D}_t\}_{t=1}^T$, each representing observations at time $t$, either as histograms or via kernel density estimates.
  \item A dictionary of basis functions $\{\psi_j\}_{j=1}^J$ spanning a subspace of $L^2(p_s)$. Hermite polynomials or kernel-based bases are typical choices.
\end{itemize}

\subsection{Stationary Density Estimation}

When $p_s$ is not known analytically, it can be estimated using kernel density estimation (KDE):
\[
\widehat{p}_s(x) = \frac{1}{M h^d} \sum_{k=1}^M K\left(\frac{x - z_k}{h}\right),
\]
where $K$ is a smooth symmetric kernel and $h > 0$ is a bandwidth parameter chosen via cross-validation or rule-of-thumb methods.

\subsection{Algorithm Description}

The DPDD procedure consists of the following steps:

\begin{enumerate}[label=\emph{Step \arabic*.}, leftmargin=20pt]
  \item Estimate the stationary density $\widehat{p}_s$ and compute importance weights $w_k \propto \widehat{p}_s(z_k)$.

  \item Construct the weighted EDMD matrices:
  \[
  G_M = \sum_{k=1}^M w_k \Psi(z_k)\Psi(z_k)^{\top}, \quad
  A_M = \sum_{k=1}^M w_k \Psi(z_{k+1})\Psi(z_k)^{\top}.
  \]

  \item Compute the Koopman operator estimate $\widehat{K}_M = A_M G_M^\dagger$ and obtain its leading $r$ eigenpairs $\{(\widehat\mu_j, \widehat\xi_j)\}$.

  \item Extract coefficients $c_j(T)$ by projecting the most recent distribution $\mathcal{D}_T$ onto the Koopman basis.

  \item Propagate the coefficients forward:
  \[
  c_j(T + h) = e^{\widehat{\lambda}_j h} c_j(T), \quad \text{where } \widehat{\lambda}_j = \log \widehat{\mu}_j.
  \]

  \item Reconstruct the future density using:
  \[
  \widehat{p}_{T+h}(x) = \widehat{p}_s(x) \left( 1 + \sum_{j=1}^r c_j(T+h) \widehat{\varphi}_j(x) \right).
  \]
\end{enumerate}

\subsection{Remarks on Implementation}

The computational cost is dominated by the eigen-decomposition of the operator $\widehat{K}_M$, which has a complexity $\mathcal{O}(J^3)$. For moderate values of $J$ (e.g., $J \leq 100$), this is feasible by using standard hardware. Techniques such as randomized SVD or kernel-based feature approximations may further accelerate computation.

\subsection{Relation to Functional and Kernel Methods}

DPDD offers a geometry-aware forecasting framework that seamlessly integrates with functional data analysis and kernel methods. Unlike traditional functional autoregression models, DPDD bypasses tangent-space projections and utilizes modal extrapolation in Wasserstein space.  Thus, it is valid for multivariate settings. Basis $\{\psi_j\}$ can be chosen from

\begin{itemize}
  \item orthogonal polynomials (e.g., Hermite functions),
  \item kernel eigenfunctions from Reproducing Kernel Hilbert Spaces (RKHS),
  \item random Fourier features for high-dimensional scalability \citep{rahimi2007random}.
\end{itemize}

This flexibility links DPDD directly to the themes of functional data and kernel-based modeling, as emphasized in this special issue.

\section{Simulation Studies}
\label{sec:simulation}

We assessed the performance of DPDD on simulated time series with known transition dynamics and compared it with benchmark methods, including Wasserstein Autoregression (WAR) and classical autoregressive models. Simulations were performed with $n_{\mathrm{exp}} = 500$ repetitions. Here, $M$ represents the length of the training trajectory.

In the first setup, we compared the prediction error in different stochastic systems between DPDD and WAR in the one-dimensional case and other baselines in the two-dimensional case because WAR is not valid for multi-variate case. In the second setup, we compared the prediction error between the weighted DPDD, unweighted DPDD and WAR, where the two results show a comparison with different $M$. In the third setup, we examined the long-term extrapolation predictive ability of DPDD. A simulation of the local stationary case was included in the first setup using an AR(2) process with coefficients drifting with time.

Both DPDD and WAR require the stochastic system to be stationary. However, in real life, we may encounter stochastic systems that are not stationary. To address this, we introduce the Sliding-window DPDD, which deals with the local stationary case in which the assumption is not seriously violated.

\subsection{Locally Stationary Setting: Sliding-Window DPDD}

The original DPDD assumes a globally stationary density $p_s$. However, real-world time series often exhibit \emph{local stationarity} in which distributions evolve slowly over time. To accommodate this, we propose a \emph{sliding-window} variant of DPDD (SW-DPDD), which re-estimates the Koopman eigenspace in a rolling window of length $W$.

Let $\{\mu_t\}_{t=1}^T$ denote the empirical distributions observed at evenly spaced times and let $\Phi(\cdot)$ denote the feature map associated with the chosen basis. For each time $t \geq W$, we define a local block $\mathcal{D}_t = \{\mu_{t-W+1}, \dots, \mu_t\}$ and proceed as follows:

\begin{enumerate}[leftmargin=20pt]
  \item Estimate the local stationary density $\widehat{p}_s^{(t)}$ via KDE on pooled samples in $\mathcal{D}_t$.
  \item Compute weighted Gram and cross-covariance matrices $G_t$, $A_t$ using local weights $w_k \propto \widehat{p}_s^{(t)}(x_k)$.
  \item Compute $K_t = A_t G_t^{\dagger}$ and extract the leading $N_{\text{mode}}$ eigenpairs.
  \item Project $\mu_t$ onto the eigenfunctions to obtain coefficients $c_i^{(t)}$.
  \item Predict $c_i^{(t+h)} = c_i^{(t)} e^{\lambda_i^{(t)} h}$ and reconstruct $\widehat{p}_{t+h}$.
\end{enumerate}

A practical rule for selecting $W$ is related to the empirical mixing time $\tau_{\text{mix}}$ from autocorrelation of the summary statistics. We used $W \in [2\tau_{\text{mix}}, 5\tau_{\text{mix}}]$, tuned via cross-validation over the training set.

\subsection{Simulations Setup}

First, we explored the predictive ability among different scenarios. We fixed a trajectory  with a length of $T = 20$ time points. The forecast performance was evaluated on the hold-out window $\mathcal{H} = \{T_0+1, \dots, T\}$ with $T_0 = \lfloor 0.7T \rfloor$. Therefore, we used the following data-generating processes:

\begin{enumerate}
  \item AR(1): $X_t = 0.9 X_{t-1} + \varepsilon_t$, $\varepsilon_t \sim \mathcal{N}(0, 0.49)$
  \item OU process: $\mathrm{d}X_t = -X_t \, \mathrm{d}t + 0.7 \, \mathrm{d}W_t$, $\Delta t = 0.01$
  \item AR+OU: additive combination of AR(1) and OU processes
  \item AR(2): $X_t = 0.6 X_{t-1} + 0.2 X_{t-2} + \varepsilon_t$, $\varepsilon_t \sim \mathcal{N}(0, 0.49)$
  \item 2D OU: two independent OU processes in $\mathbb{R}^2$
\end{enumerate}

Second, we changed the length of the training trajectory for a non-linear OU process: $\mathrm{d}X_t = (-aX_t+bX_t^3) \, \mathrm{d}t + c \, \mathrm{d}W_t$, $\Delta t = 0.01$ to test whether DPDD could capture the complex features of a stochastic system. Here, we set $a=0.5, b=0.1, c=0.7. $ And the length of the training trajectory $M$ was chosen from $256, 512, 1024, 2048$. Five hundred different trajectories were generated independently from the given stochastic dynamic systems, with the first half being the training set and the other half being the testing set.

In the third setup, we compared different extrapolation steps. Here, non-linear OU process: $\mathrm{d}X_t = (-aX_t+bX_t^3) \, \mathrm{d}t + c \, \mathrm{d}W_t$, $\Delta t = 0.01$ and AR(2): $X_t = 0.6 X_{t-1} + 0.2 X_{t-2} + \varepsilon_t$, $\varepsilon_t \sim \mathcal{N}(0, 0.49)$ were used to test DPDD in linear and non-linear processes. The extrapolation steps are chosen from 1, 2, 4, 8, 16, 32, 64, 128, 256, 512, and 1024.

We compared DPDD with WAR \citep{chen2023wasserstein} and classical VAR baselines for the AR models. The forecast error was measured using the empirical 2-Wasserstein mean-squared error (MSE):
\[
\mathrm{MSE}_{W_2} = \frac{1}{|\mathcal{D}_{\mathrm{test}}|} \sum_{x \in \mathcal{D}_{\mathrm{test}}} W_2^2(x, \widehat{p}).
\]

\begin{figure}[htbp]
  \centering
   \includegraphics[width=\linewidth]{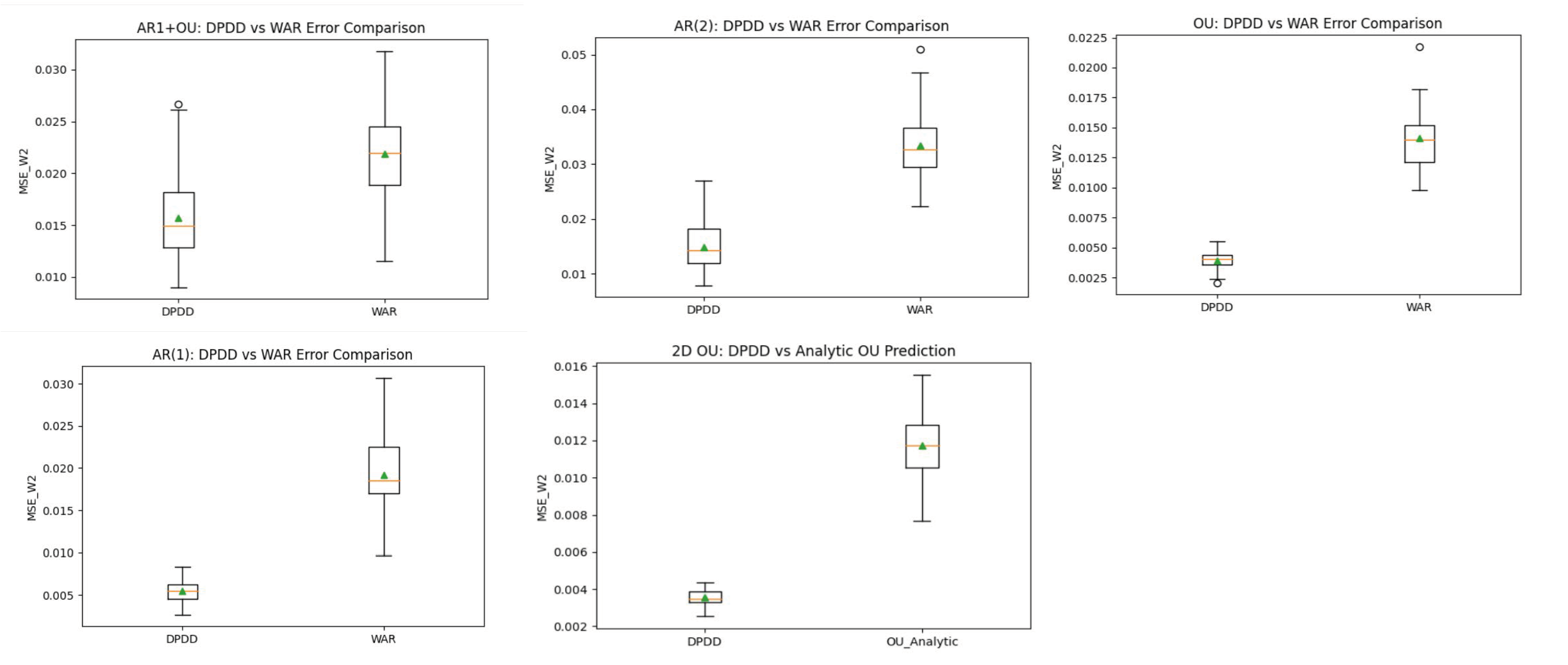}
  \caption{The boxplot of the result of the stationary case in the simulation experiment in Section \ref{sec:simulation}. The prediction error is evaluated by Wasserstein discrepancy. In all scenarios DPDD methodology attains the least $\mathrm{MSE}_{W_2}$}
  \label{fig:boxplot}
\end{figure}

\begin{table}[htbp]
  \centering
  \centering
  \caption{Error ($\mathrm{MSE}_{W_2}$) for each model in the simulation setup.}
  \begin{tabular}{|c|c|c|c|c|c|}
    \hline
     & AR(1) & AR(2) & OU & 2D OU & Mixing AR(1)+OU \\ \hline
    
    DPDD & 0.009 & 0.014 & 0.004 & 0.004 & 0.015 \\ \hline
    
    WAR & 0.019 & 0.033 & 0.015 & 0.012 & 0.022 \\ \hline

  \end{tabular}
  \vspace{5pt} 
  
  \label{tab:sim}
\end{table}

\begin{figure}[htbp]
  \centering
   \includegraphics[width=0.7\linewidth]{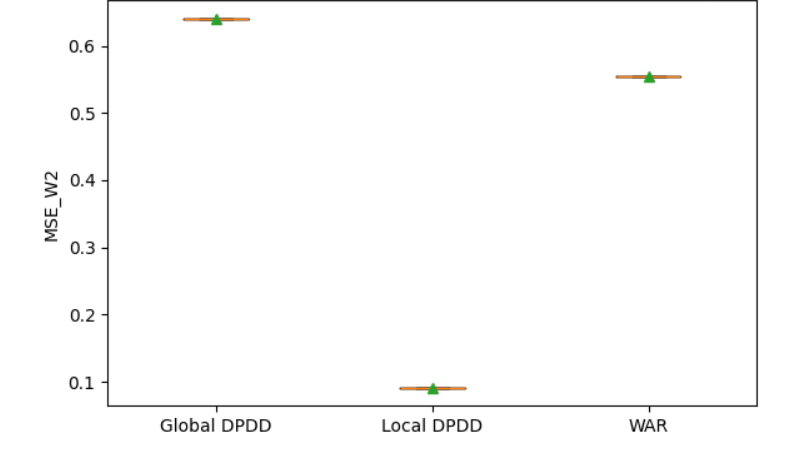}
  \caption{The boxplot of the result of the locally stationary case in the simulation experiment in Section \ref{sec:simulation}. The $\mathrm{MSE}_{W_2}$ of global DPDD is 0.642, that of WAR is 0.570, and that of sliding-window DPDD is 0.095.}
  \label{fig:figure3}
\end{figure}

\begin{figure}[htbp]
  \centering
   \includegraphics[width=0.7\linewidth]{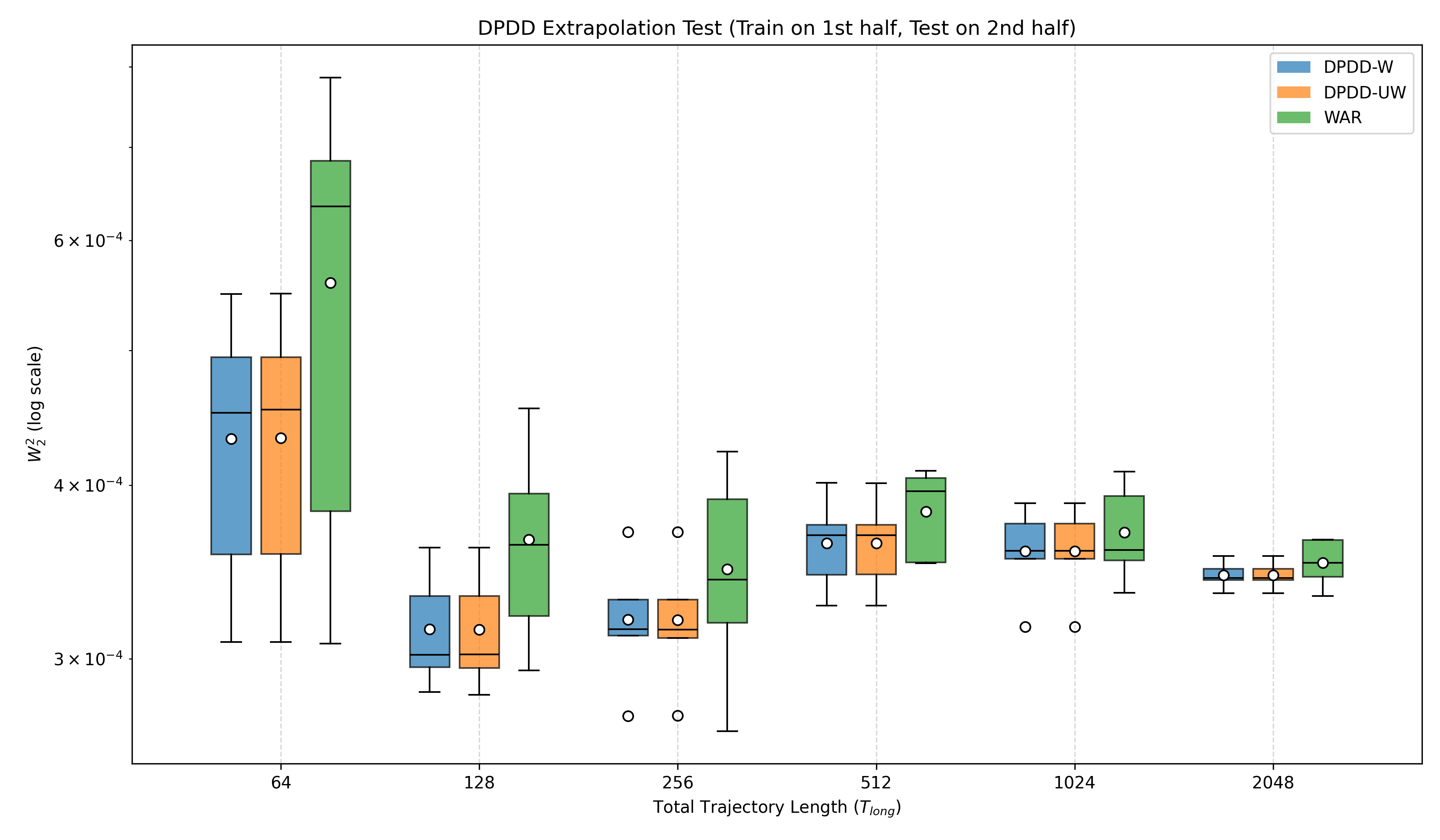}
  \caption{The boxplot of the result of the second setup in the simulation experiment in Section \ref{sec:simulation}. The length of the training trajectory $M$ varies. In general, DPDD outperforms WAR, which shows that DPDD successfully capture the dynamic feature of the stochastic system. Weighted DPDD is more accurate when the training trajectory is short, while there is no significant difference between weighted and unweighted DPDD when the training trajectory is long.}
  \label{fig:figure4}
\end{figure}

\begin{figure}[htbp]
  \centering
   \includegraphics[width=\linewidth]{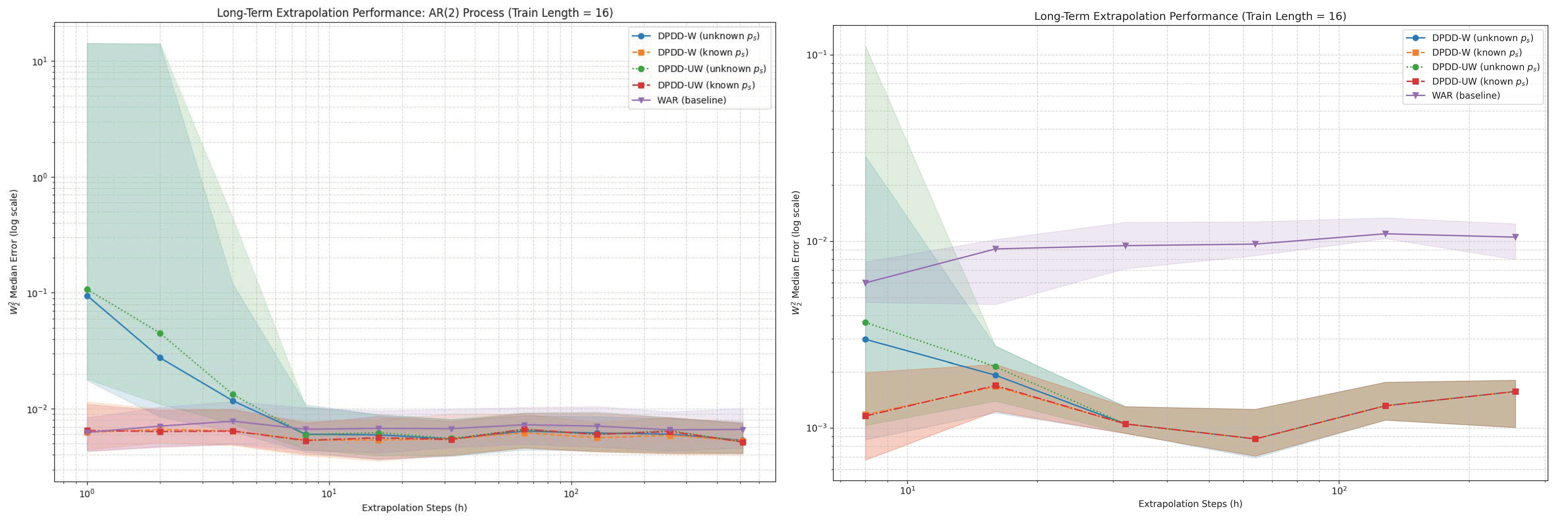}
  \caption{The boxplot of the third setup of the locally stationary case in the simulation experiment in Section \ref{sec:simulation}. The graph in the left is AR(2) process and the other is non-linear OU process. The prediction error of WAR increases when the extrapolation step $h$ is increasing, meanwhile DPDD achieves a good performance in long-term extrapolation in both settings.}
  \label{fig:figure5}
\end{figure}

\subsection{Results}

In the first setup, DPDD consistently yields the lowest $\mathrm{MSE}_{W_2}$ across all stationary scenarios, especially when the true dynamics admit a low-dimensional Koopman representation. WAR underestimates tail behavior because of its truncation of principal components, whereas classical baselines fail to capture distributional features entirely.

In a locally stationary setting, global DPDD suffers from performance degradation; however, the SW-DPDD variant substantially mitigates this, demonstrating its robustness to nonstationary drift. DPDD performed worse than WAR in local stationary experiments. In the case of local stationary conditions, the SW-DPDD should be used. Full error statistics are reported in Appendix~\ref{fig:boxplot}.

If the system experiences significant drift and does not exhibit local stationary characteristics, the theoretical foundation of DPDD—specifically, the previously mentioned decomposition of the probability density—will no longer be valid. Discrete eigenvalues may become a continuous spectrum, and the mode extraction becomes unstable. Furthermore, the extracted modes may reflect noise rather than true distribution changes, which significantly degrades the prediction performance.

In the second setup, we can see that DPDD performs better across all $M$ settings. However, as $M$ increases, the prediction error does not decrease consistently, owing to the KDE error.

In the third setup, we can see that in both stochastic systems, for long-term extrapolation prediction,  WAR suffers from step-by-step error accumulation, while DPDD has a good performance. At the same time, compared to the weighted DPDD, unweighted DPDD performs relatively poorly when the number of extrapolation steps is small; however, this gap will be gradually eliminated in long-term forecasts. Meanwhile, awareness of the stationary distribution will be helpful for short-term extrapolation, although for long-term prediction, knowing the stationary distributions will not be significant.

It can be seen from the experimental results that if the stationary distribution is known in advance or there is at least an accurate description of the stationary distribution, the prediction accuracy is significantly improved compared to when nothing is known about the stationary distribution (in this case, the stationary distribution must be estimated by KDE). The extrapolation experimental results also show that if the KDE estimation is not sufficiently accurate, it will lead to error accumulation in the long-term extrapolation estimation.

Figure~\ref{fig:figure5} compares the $W_2$ forecast error of the weighted/unweighted DPDD and WAR for the forecast horizons $h \in {1, 2, 4, 8, \dots, 1024}$. WAR’s error increases sharply with horizon length owing to step-by-step error accumulation, whereas DPDD maintains a stable accuracy for long horizons. Weighted DPDD shows clear advantages for short horizons; however, this gap narrows as the horizon grows.

\section{Real-World Application: U.S. Housing Market Forecasting}
\label{sec:app}

We analyze U.S. housing price distributions using Zillow’s \emph{Metro Median Sale Price} panel to demonstrate the practical utility of DPDD in a real-world functional time series context. This dataset comprises $m = 445$ U.S. metropolitan statistical areas (MSAs) with monthly inflation-adjusted median home prices from January 2008 to March 2025 ($T = 207$ time points). See \url{https://www.zillow.com/} for further detail.

Each month-specific vector $(p_{1t}, \dots, p_{mt})$ is normalized by its mean to form a relative price distribution as follows:
\[
\mu_t := \text{Empirical distribution of } \left( p_{1t}, \dots, p_{mt} \right) / \left( \tfrac{1}{m} \sum_{i=1}^m p_{it} \right).
\]

We split the data chronologically from January 2008 to December 2019 ($T_{\text{train}} = 144$) for training and from January 2020 to March 2025 ($T_{\text{test}} = 63$) for out-of-sample forecasting.

\subsection{Forecasting Procedure}
We compare two methods:
\begin{itemize}
  \item{Wasserstein Autoregression (WAR)}: Applies functional principal component analysis (FPCA) to the quantile functions $Q_t(u) = F_{\mu_t}^{-1}(u)$, retains components explaining 95\% of variance, and fits an AR(1) model to the FPCA scores using the GetProj–log correction of \citep{chen2023wasserstein}.
  \item{DPDD}: Estimates a global stationary density $\hat{p}_s$ via KDE over the training period. A polynomial EDMD is constructed using degree-2 basis functions, and the top two Koopman modes are used to forecast future modal coefficients: $\mathbf{c}_{t+h} = e^{\Lambda h} \mathbf{c}_t$.
\end{itemize}

For each test month, we compute the squared 2-Wasserstein distance between the predicted and empirical distributions:
\[
W_2^2(\hat\mu_t, \mu_t) = \int_0^1 \left( \hat{Q}_t(u) - Q_t(u) \right)^2 \, du,
\]
and the average error over all $T_{\text{test}}$ months.

\begin{figure}[htbp]
\centering
\includegraphics[width=0.7\linewidth]{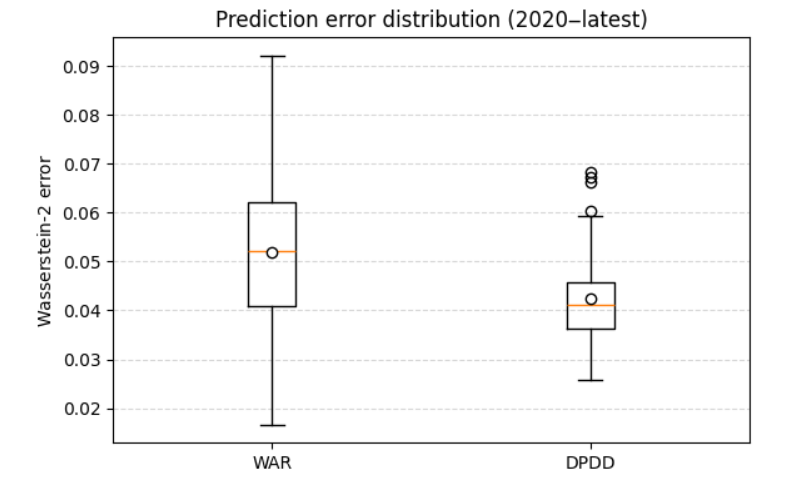}
\caption{Prediction errors for U.S. housing price distributions. DPDD achieves lower average Wasserstein MSE (0.042) than WAR (0.052), with greater robustness during volatile periods.}
\label{fig:boxplots}
\end{figure}

\begin{figure}[htbp]
\centering
\includegraphics[width=0.9\linewidth]{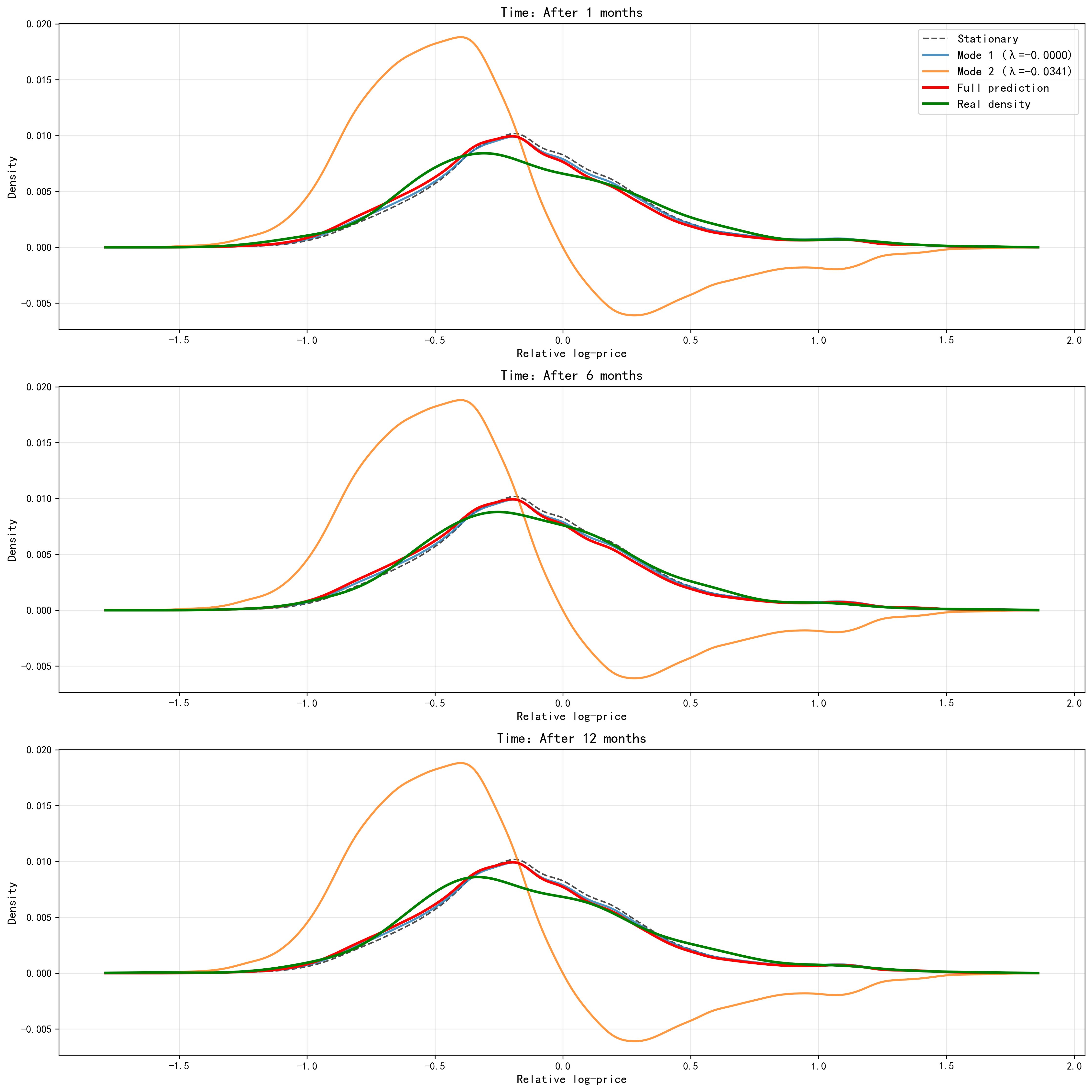}
\caption{The two major modes for the US house price data. Mode 1 captures the steady-state structure of the overall system, and Mode 2 represents the dynamic characteristics of the system.}
\label{fig:boxplots_modes}
\end{figure}

\begin{figure}[htbp]
\centering
\includegraphics[width=0.7\linewidth]{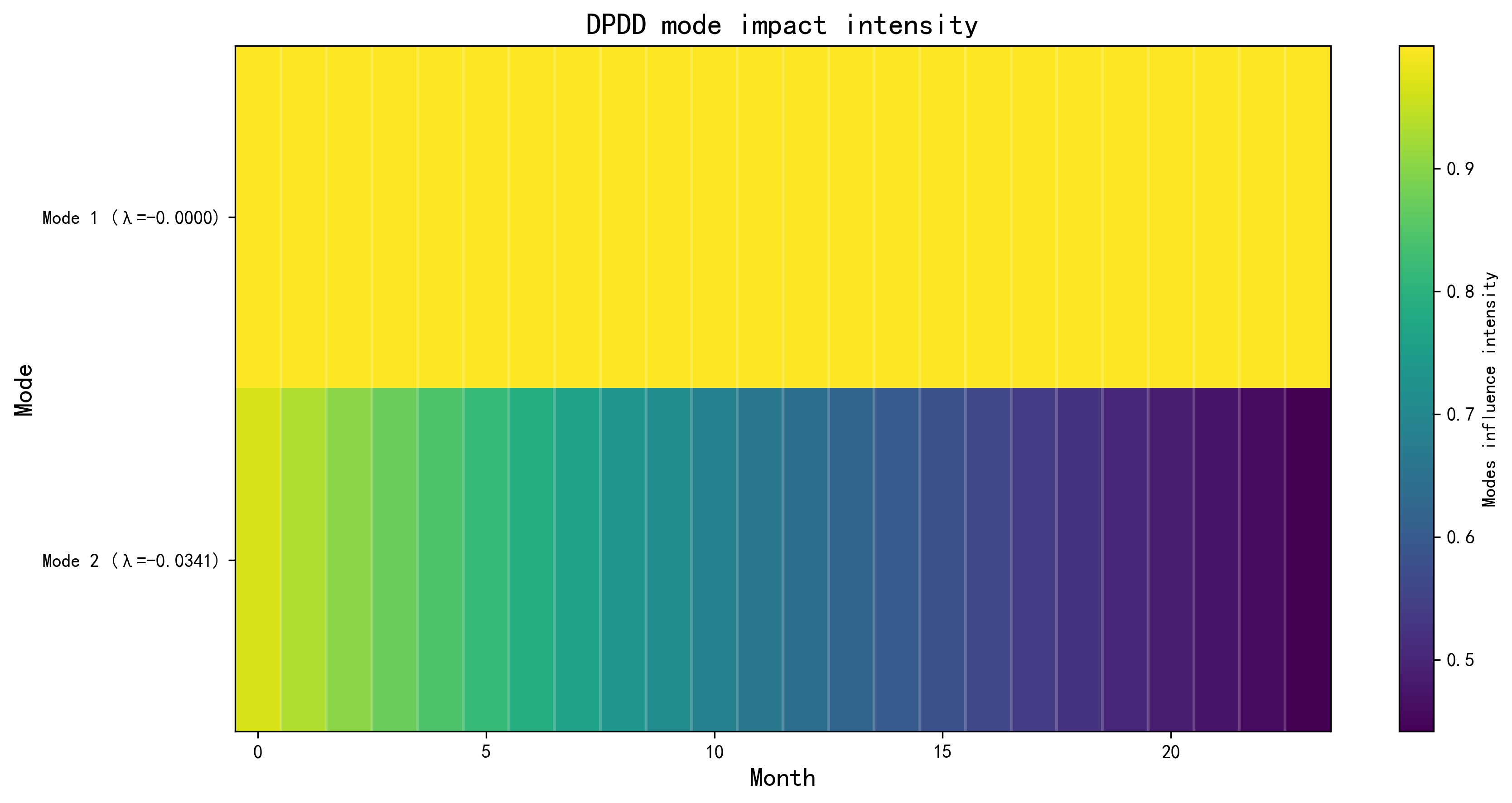}
\caption{The impact intensity of the modes in US house price data prediction. Mode 1 represents a stationary distribution, showing stable characteristics for all months; Mode 2 represents a fluctuating impact, showing different impacts on different forecast months.}
\label{fig:boxplots_heat}
\end{figure}

\subsection{Discussion of Results}

The mean $W_2$ error of DPDD is 0.042, which improves upon the WAR’s 0.052 over the test set. The advantage of DPDD is particularly evident during periods of high volatility (e.g., 2021–2022) when market-wide shifts affect all MSAs simultaneously. In contrast, WAR sometimes performs better during abrupt local corrections (e.g., April 2020 or November 2022) because of its short-horizon adaptation.

Mode 1 reflects the underlying structure of housing price distribution in the United States: some cities are consistently more expensive than others. Mode 2 reflects market volatility: how prices move from one range to another during the economic cycle.

DPDD benefits from:
\begin{itemize}
  \item capturing the smooth, long-term evolution of spatial distributions via Koopman modes;
  \item exploiting the global structure of the density manifold rather than local linearization;
  \item maintaining stability in low-variance months, where WAR's FPCA-based scores become noisy.
\end{itemize}

These findings highlight the suitability of DPDD for medium-to-long-term distributional forecasting in functional economic time series, where interpretability and robustness are critical.

In the U.S. house price data experiment, the construction of $G$ and $A$ matrix cost $O(MJ^2)$, whereas SVD costs $O(J^3)$ theoretically. In our experiment, the SVD part required 19.0583 s when $J$ was 2 and  56.7727 s when $J$ was 3.

\section{Discussion}
\label{sec:discussion}

This study presents a functional time series forecasting method specifically designed for probability distributions. By utilizing the spectral decomposition of the Koopman operator and integrating it with optimal transport geometry, the proposed DPDD framework provides an interpretable and closed-form approach to modeling distributional dynamics.

The idea of DPDD originates from the work of \citet{zhao2022dpdd}; however, unlike their use of DPDD to characterize diffusion in complex systems, we focus on extending DPDD for forecasting time series on the Wasserstein manifold and provide the corresponding theoretical foundation. The experimental results demonstrate that DPDD outperforms WAR in both long-term extrapolation and real-data prediction, and it can be naturally applied to multidimensional distributions. This discussion highlights the implications of our findings and suggests future research directions at the intersection of functional data analysis and kernel-based methods.

In higher dimensions ($d > 1$), classical Wasserstein-based methods like WAR face limitations due to the lack of explicit multivariate quantile functions and a canonical tangent space structure. These limitations hinder both interpretability and scalability. In contrast, DPDD addresses these challenges by learning a spectral basis directly from trajectory data, eliminating the need for log-exp maps or tangent space projections. This capability is especially beneficial in multivariate contexts where traditional functional principal components can become ill-posed or computationally expensive.

The basis expansion in DPDD can flexibly incorporate polynomial, kernel, or data-driven features. For example, tensor-product Hermite polynomials and random Fourier features offer practical choices in high-dimensional settings. This flexibility aligns well with the goals of kernel-based learning and functional approximation, reinforcing compatibility of the DPDD with the themes of this special issue.

The effectiveness of the traditional EDMD depends heavily on the choice of basis (dictionary) functions, and the computational cost increases substantially with larger dictionary sizes $J$. According to  theoretical results, the Koopman operator can only be accurately approximated when the dictionary is sufficiently rich. However, in the systems considered in this study, Hermite or Chebyshev polynomials up to the order of four were sufficient to capture most of the nonlinear features. Fourier bases are effective for systems with periodic dynamics. An excessively large $J$ may cause overfitting and numerical instability, whereas too small a $J$ may fail to capture the important dynamic modes. These observations provide practical guidance for selecting both the type and size of the basis  DPDD.

It is worth noting that when the sample sequence is small, extreme cases can have a certain impact on the training; therefore, truncating small weights is beneficial. Furthermore, certain dictionary functions generate negative probability densities. Therefore, when using the dictionary function for projection, logarithmic density can be used to avoid negative probability densities.

The choice of the number of modes also has a non-negligible impact on the results. The eigenvalues of the Koopman matrix need not be real; their moduli (the absolute values in the real case) represent the energy of the modes. Our selection rule is to truncate where a pronounced spectral gap is observed—that is, when the eigenvalue magnitudes drop sharply.\\
If the system significantly deviates from stationarity or if the stationary density $p_s$ is inaccurately estimated, the theoretical basis of DPDD may break down. In these situations, discrete Koopman eigenvalues can converge into a continuous spectrum, making mode extraction unstable. Consequently, the resulting modes may capture sampling noise instead of true distributional changes, which can lead to considerable forecast degradation. While SW-DPDD alleviates some of these problems by locally updating the baseline, additional strategies—such as enhancing the accuracy of $p_s$ estimation (e.g., more data, refined KDE) or using robust spectral estimation techniques—can further improve resilience to these challenges. As for KDE-related parameter selection (e.g., bandwidth), please refer to \citep{silverman1986density}.

In terms of applications, DPDD is well-suited for domains in which the object of interest evolves as a distribution rather than a point. These include financial risk modeling (e.g., evolving return distributions), climate science (e.g., spatial temperature fields), and neuroimaging (e.g., activation maps across time). Beyond forecasting distributional time series, DPDD also has the potential to capture the dynamical characteristics of parameter distributions arising during the training of deep neural networks. The ability to forecast entire distributions in a geometry-consistent manner offers not only predictive power, but also interpretability, which is an important consideration in scientific applications.

Several extensions of this framework merit further attention. One important approach is adaptive basis selection, in which a dictionary is learned from the data to better capture the dominant modes of variability. Another promising avenue is the development of online or streaming variants of DPDD that can update the modal coefficients in real time. Furthermore, a rigorous analysis of the approximation error introduced by kernel-based or neural approximations of Koopman eigenfunctions remains an open challenge. Finally, hybrid models that combine the localized interpretability of WAR with the global spectral forecasting capability of DPDD offer the best results for both worlds.

Overall, DPDD provides a statistically grounded and geometrically reliable tool for distributional timeseries analysis. The integration of functional approximation, optimal transport, and spectral operator theory places it at the intersection of several active areas of statistics and machine learning, making it a strong candidate for future exploration and applications.

\section*{Acknowledgments}
We sincerely thank the two reviewers for their thoughtful and constructive comments, which have greatly improved the quality of the manuscript. This work was partially supported by JSPS KAKENHI Grant Numbers 25H01464 and 23K28042. The authors thank Dr. Masaaki Imaizumi and Dr. Ryo Okano for their insightful feedback on parts of this study. 

\section*{Data Availability}
The data used in the experiments in Section \ref{sec:app} can be find on 
\url{https://www.zillow.com/}.

\appendix
\section{Notation and Standing Assumptions}\label{app:notation}

We recap all symbols that enter the proofs and explicitly list the technical
conditions under which the main results hold.

\subsection*{A.1  Symbols}

\begin{tabular}{@{}ll@{}}\toprule
Symbol & Description \\ \midrule
$\mathcal{S}$                & compact support of the stationary density $p_{\mathrm{s}}$ \\
$\{\psi_{j}\}_{j\ge1}$       & $L^{2}(p_{\mathrm{s}})$–orthonormal dictionary          \\
$J$                          & number of basis functions retained                     \\
$\Psi_{J}(x)$                & $(\psi_{1}(x),\dots,\psi_{J}(x))^{\mathsf{T}}$          \\
$G_{M},A_{M}$                & weighted EDMD matrices, Eq.~\eqref{eq:GA_mats}          \\
$G_{\ast},A_{\ast}$          & population analogues of $G_{M},A_{M}$                  \\
$\varphi_{i},\lambda_{i}$    & true Koopman eigenfunctions and eigenvalues            \\
$\xi_{i},\mu_{i}$            & projections $\Pi_{J}\varphi_{i}$ and $e^{\lambda_{i}\Delta t}$\\
$\widehat\xi_{i},\widehat\mu_{i}$ & empirical eigenpairs of $A_{M}G_{M}^{\dagger}$   \\
$c_{i}(t)$                   & modal coefficient in expansion~\eqref{eq:q_expansion}  \\
$\widehat{c}_{i}(t)$         & estimated modal coefficient                            \\ \bottomrule
\end{tabular}

\subsection*{A.2  Assumptions}
\begin{enumerate}[label=\textbf{(A\arabic*)},leftmargin=20pt,itemsep=4pt]
\item\label{assump:KDE1}  \textbf{Density smoothness}
$p_s \in C^2(S)$ and $\sup_{x\in S}\|\nabla^2 p_s(x)\| \le C_p$.
\item\label{assump:KDE2}  \textbf{Kernel properties}
$K$ is bounded, symmetric, of order 2, $\int K = 1$, 
                 and Lipschitz with constant $L_K$.
\end{enumerate}
\begin{enumerate}[label=\textbf{(B\arabic*)},leftmargin=20pt,itemsep=4pt]
\item\label{B:ergodic} \textbf{Ergodic diffusion}  
      The SDE~\eqref{eq:sde} is ergodic with unique stationary density
      $p_{\mathrm{s}}$; its trajectory is exponentially $\beta$--mixing:
      $\beta(\tau)\le C_{0}e^{-\kappa\tau}$.

\item\label{B:dictionary} \textbf{Dictionary completeness}  
      There exist constants $C_{\varphi},q>0$ such that for each leading mode
      $\varphi_{i}$ ($i=1,\dots,r$),
      $\|(\mathrm{Id}-\Pi_{J})\varphi_{i}\|_{L^{2}(p_{\mathrm{s}})}
          \le C_{\varphi}J^{-q}$.

\item\label{B:bounded} \textbf{Moment bound}  
      $\sup_{x\in\mathcal{S}}\|\Psi_{J}(x)\|^{4}\le C_{\Psi}$ for all $J$.

\item\label{B:gram} \textbf{Population non--degeneracy}  
      The Gram matrix
      $G_{\ast}=\mathbb{E}_{p_{\mathrm{s}}}[\Psi_{J}(X)\Psi_{J}(X)^{\top}]$
      is positive definite with $\lambda_{\min}(G_{\ast})\ge\lambda_{0}>0$.

\item\label{B:gap} \textbf{Spectral gap}  
      $\min_{1\le i\le r}(\lambda_{i-1}-\lambda_{i})\ge\gamma>0$ with
      $\lambda_{0}=0$.

\item\textbf{Dictionary smoothness} 
Each basis function $\psi_j$ is $L_{\Psi}$–Lipschitz on $S$:
\[
  |\psi_j(x)-\psi_j(y)| \le L_{\Psi}\,\|x-y\|, 
  \qquad \forall x,y\in S,\; j=1,\dots,J.
\]
\label{assump:dict-lip}
\end{enumerate}
Conditions~\ref{B:ergodic}--\ref{B:bounded} guarantee concentration of
empirical moments; \ref{B:gram} ensures invertibility of $G_{M}$ for large
$M$; \ref{B:gap} stabilizes eigenvectors under perturbations.

\section{Auxiliary Lemmas}\label{app:lemmas}

\begin{lemma}[Concentration of Gram and Cross Matrices]\label{lem:GA_conc}
Under \textbf{\ref{B:ergodic}--\ref{B:bounded}} such that for every $t>0$
\[
  \Pr\!\{\|G_M-G_\ast\|_2\ge t\}\;\le\;
      J\exp\!\bigl(-cMt^2\bigr),\qquad
  \Pr\!\{\|A_M-A_\ast\|_2\ge t\}\;\le\;
      J\exp\!\bigl(-cMt^2\bigr).
\]
Consequently\[
  \mathbb{E}\|G_{M}-G_{\ast}\|_{2}=O(M^{-1/2}),\qquad
  \mathbb{E}\|A_{M}-A_{\ast}\|_{2}=O(M^{-1/2}).
\]
\end{lemma}

\begin{proof}

Applying the matrix Bernstein inequality for geometrically $\beta$–mixing random matrices in \citep{BannaMerlevedeYoussef2016} to the centered summands $Y_k:=w_k\Psi_J(z_k)\Psi_J(z_k)^\top-G_\ast$ whose bounded spectrum and variance are guaranteed by Assumptions~\ref{B:ergodic}--\ref{B:bounded} yields the tail bound $\|G_M-G_\ast\|_2,\|A_M-A_\ast\|_2=O_{\mathbb{P}}(M^{-1/2})$, proving the lemma.

\end{proof}

\begin{lemma}[Operator Perturbation]\label{lem:op_perturb}
Let $\mathcal{K}^{\Pi}_{\ast}=A_{\ast}G_{\ast}^{-1}$ and
$\widehat{\mathcal{K}}_{M}=A_{M}G_{M}^{\dagger}$.
Under Assumptions \textbf{\ref{B:ergodic}--\ref{B:gram}},
\[
  \|\widehat{\mathcal{K}}_{M}-\mathcal{K}^{\Pi}_{\ast}\|_{2}
  = O_{\mathbb{P}}(M^{-1/2}).
\]
\end{lemma}

\begin{proof}
Write $\widehat{\mathcal{K}}_{M}-\mathcal{K}^{\Pi}_{\ast}
 =(A_{M}-A_{\ast})G_{M}^{\dagger}
  +A_{\ast}(G_{M}^{\dagger}-G_{\ast}^{-1})$.
Lemma~\ref{lem:GA_conc} bounds the first term.
For the second,
$G_{M}^{\dagger}-G_{\ast}^{-1}=G_{\ast}^{-1}(G_{\ast}-G_{M})G_{M}^{\dagger}$
and use $\|G_{M}^{\dagger}\|_{2}\le2/\lambda_{0}$ for large $M$.
\end{proof}
\begin{lemma}\label{lem:phi-lip}
Under Assumption~\ref{assump:dict-lip}, each coordinate of $\Phi(x)=\Xi^\top\Psi(x)$ is Lipschitz on the support of $p_{\mathrm{s}}$.
\end{lemma}

\begin{lemma}[Uniform KDE error]\label{lem:kde}
Let $\hat p_s$ denote the kernel density estimator with bandwidth
$h_M=M^{-\beta}$, where $\beta\in(0,1/2)$.  Under
Assumptions~\ref{assump:KDE1}--\ref{assump:KDE2} show that
\[
  \|\hat p_s-p_s\|_\infty \;=\;
  O_P\!\bigl(M^{-\beta}\bigr)
  \;+\;
  O_P\!\bigl(M^{-(1-\beta)/2}\bigr).
\]
Choosing any $\beta\in(\tfrac13,\tfrac12)$ yields
$\|\hat p_s-p_s\|_\infty=o_P(M^{-1/2})$.
\end{lemma}

\section{Spectral Convergence Theorem}\label{app:spectral_proof}

\begin{theorem}[Spectral Convergence of EDMD]\label{thm:spec_conv_app}
Suppose \textbf{\ref{B:ergodic}--\ref{B:gap}} and
$J\ll M^{1/2}$.  Then
\[
  \max_{1\le i\le r}
   \Bigl(|\widehat\mu_{i}-\mu_{i}|
        +\|\widehat\xi_{i}-\xi_{i}\|_{2}\Bigr)
   = O_{\mathbb{P}}(M^{-1/2}+J^{-q}).
\]
\end{theorem}

\begin{proof}
By Lemma~\ref{lem:op_perturb},
$\|\widehat{\mathcal{K}}_{M}-\mathcal{K}^{\Pi}_{\ast}\|_{2}
  =O_{\mathbb{P}}(M^{-1/2})$.
Add the deterministic Galerkin bias
$\|\mathcal{K}^{\Pi}_{\ast}-\mathcal{K}_{\ast}\|_{2}=O(J^{-q})$
based on \textbf{\ref{B:dictionary}}.
Davis–Kahan applied a gap $\gamma$
(Assumption~\textbf{\ref{B:gap}}) yields the claim.
\end{proof}

\section{Forecast Error in $W_{2}$}\label{app:w2_proof}

\begin{theorem}[Finite–h prediction error in $W_2$]\label{thm:W2-risk}
Let assumptions \ref{B:ergodic}--\ref{assump:dict-lip} and the additional conditions
\ref{as:C1}--\ref{as:C3} below hold:
\begin{enumerate}[label=\textbf{(C\arabic*)}]
  \item\label{as:C1}  All densities lie on a common compact set
        $S\subset\mathbb{R}^{d}$ with diameter $D$ and volume $\operatorname{vol}(S)$.
  \item\label{as:C3}   For $j>r$ the Koopman
        eigenvalues satisfy $\Re\,\lambda_j\le -\eta<0$.
\end{enumerate}
Without a loss of generality, we assume that the functions
        $\{\psi_j\}_{j=1}^{J}$ are orthonormal in $L^2(p_s)$. Fix a finite forecast horizon $h>0$ and assume a sample size $M\to\infty$
whereas the retained rank $r$ remains constant. Then,
\begin{equation}\label{eq:W2-bound}
  \mathbb{E}\,
  W_2^2\!\bigl(p_{T+h},\hat p_{T+h}\bigr)
  = \mathcal{O}\!\bigl(M^{-1/2}\bigr)
    + \mathcal{O}\!\bigl(e^{-\,\eta h}\bigr),
\end{equation}
where the hidden constants depend only on $D$, $\operatorname{vol}(S)$,
the dictionary Lipschitz and $L^\infty$ bounds, and the spectrum of $K_\ast$.
\end{theorem}

\begin{proof}
Let the true density and its DPDD forecast be
\[
  p_{T+h}(x)=\sum_{j\ge1}c_j(h)\psi_j(x),\qquad
  \hat p_{T+h}(x)=\sum_{j=1}^{r}\hat c_j(h)\psi_j(x).
\]
Introduce the decomposition
\[
  \Delta(x)\;=\;\hat p_{T+h}(x)-p_{T+h}(x)
      \;=\;\underbrace{\sum_{j=1}^{r}\!\bigl[\hat c_j(h)-c_j(h)\bigr]\psi_j(x)}_{\Delta_{\mathrm{est}}}
        +\underbrace{\sum_{j>r}c_j(h)\psi_j(x)}_{\Delta_{\mathrm{tail}}},
\]
and write $\delta c_j(h)=\hat c_j(h)-c_j(h)$.

\paragraph{Step 1 (coefficient error).}
Using the Spectral Convergence Theorem~\ref{thm:spec_conv_app} together with the diagonal action of $e^{h\Lambda}$,
\[
  \mathbb{E}\,\|\delta c(h)\|_2^{2}=\mathcal{O}\!\bigl(M^{-1}\bigr).
\]

\paragraph{Step 2 (tail truncation).}
Using assumption~\ref{as:C3},
\begin{equation}\label{eq:tail-L2}
  \|\Delta_{\mathrm{tail}}\|_2^{2}
  =\sum_{j>r}c_j^2(0)\,e^{2\Re\lambda_j h}
  \le e^{-\,2\eta h}\sum_{j>r}c_j^2(0)
  =\mathcal{O}\!\bigl(e^{-\,2\eta h}\bigr).
\end{equation}

\paragraph{Step 3 ($L^2$ control).}
Because we assumed the dictionary functions to be orthonormal,
\[
  \|\Delta\|_2^{2}
    =\|\Delta_{\mathrm{est}}\|_2^{2}+\|\Delta_{\mathrm{tail}}\|_2^{2}
    =\|\delta c(h)\|_2^{2}+\|\Delta_{\mathrm{tail}}\|_2^{2}.
\]
Taking the expectations and using Step 1 and~\eqref{eq:tail-L2},
\[
  \mathbb{E}\,\|\Delta\|_2^{2}
    =\mathcal{O}\!\bigl(M^{-1}\bigr)
     +\mathcal{O}\!\bigl(e^{-\,2\eta h}\bigr).
\]

\paragraph{Step 4 ($L^1$ conversion).}
Using Hölder on the compact set $S$ (diameter $D$, volume $V$),
\[
  \mathbb{E}\,\|\Delta\|_1
  \;\le\;\sqrt{V}\,
        \bigl[\mathbb{E}\,\|\Delta\|_2^{2}\bigr]^{1/2}
  =\mathcal{O}\!\bigl(M^{-1/2}\bigr)
   +\mathcal{O}\!\bigl(e^{-\,\eta h}\bigr).
\]

\paragraph{Step 5 ($W_2$ bound).}
For measures supported on a set of diameters $D$, according to Theorem 6.15 in \citep{villani2009optimal}, we have
\[
  W_2^{2}(\mu,\nu)
    \;\le\;D\,W_1(\mu,\nu)
    \;\le\;\frac{D^{2}}{2}\,\mathrm{TV}(\mu,\nu)
    \;=\;\frac{D^{2}}{2}\,\|p_\mu-p_\nu\|_1,
\]

Applying this with $\mu=p_{T+h}$, $\nu=\hat p_{T+h}$ and using Step 4 yields
\[
  \mathbb{E}\,
  W_2^{2}\bigl(p_{T+h},\hat p_{T+h}\bigr)
    \;\le\;\frac{D^{2}}{2}\,
           \mathbb{E}\,\|\Delta\|_1
    =\mathcal{O}\!\bigl(M^{-1/2}\bigr)
     +\mathcal{O}\!\bigl(e^{-\,\eta h}\bigr),
\]
which is exactly the same as in our theorem.
\end{proof}

\section{Implementation Remarks}\label{app:implementation}

\begin{itemize}[leftmargin=15pt,itemsep=2pt]
\item \textbf{Orthogonalizing \boldmath$\Psi_{J}$.}  
      We compute $\widetilde{\Psi}=G_{M}^{-1/2}\Psi_{J}$ so that
      $\widetilde G_{M}=I$ in floating point, improving the eigenvalue
      conditioning.
\item \textbf{Regularizing $G_{M}$.}  
      Add $\varepsilon I$ with $\varepsilon=10^{-8}$ if
      $\lambda_{\min}(G_{M})<\lambda_{0}/2$.
\item \textbf{Truncation rule.}  
      Retain all modes with
      $|\widehat\mu_{j}|\ge0.9\max_{k}|\widehat\mu_{k}|$, balancing
      variance and bias empirically.
\end{itemize}

\section{Reproducibility}

To ensure reproducibility, we provide detailed information on experimental settings and hyperparameters here. For KDE, the bandwidth was chosen following \citep{silverman1986density}. Eigenmodes were truncated at pronounced spectral gaps. Here we use 3 dynamic modes for each experiment. The dictionary function is set to be the Hermite polynomials with order 3. A regularization 0.001 is used here to get a stable Koopman spectrum.
And we use a weight cut-off of 0.25 to prevent a large influence from an extreme sample.

\bibliographystyle{plainnat}
\bibliography{refs}
\nocite{*}

\section*{Statements and Declarations}
The authors declare that they have no conflict of interest.

\end{document}